\documentclass{ws-ijfcs}

\usepackage{amsmath,amssymb}
\usepackage{enumitem}
\usepackage{float}
\usepackage{array}

\def\Que{{\mathbb Q}}

\def\fmod#1 #2{#1\ ({\rm mod}\ #2)}
\DeclareMathOperator{\quo}{quo}
\def\rela{\vartriangleleft}

\begin{document}

\title{The Critical Exponent is Computable for Automatic Sequences}

\markboth{Luke Schaeffer and Jeffrey Shallit}{The Critical Exponent is
Computable for Automatic Sequences}

\title{THE CRITICAL EXPONENT IS COMPUTABLE FOR AUTOMATIC SEQUENCES}

\author{LUKE SCHAEFFER and JEFFREY SHALLIT}
\address{School of Computer Science,
University of Waterloo,
Waterloo, ON  N2L 3G1,
Canada\\
\email{\tt l3schaef@cs.uwaterloo.ca,\ shallit@cs.uwaterloo.ca}}

\maketitle

\begin{abstract}
The {\it critical exponent} of an infinite word is defined to be the supremum
of the exponent of each of its factors.  
For $k$-automatic sequences,
we show that this critical
exponent is always either a rational number or infinite, and its
value is computable.  Our results also apply to variants of the critical
exponent, such as the {\it initial critical exponent} of
Berth\'e, Holton, and
Zamboni and the {\it Diophantine exponent} of Adamczewski and Bugeaud.
Our work generalizes or recovers
previous results of Krieger and others, and  is applicable to
other situations; e.g., the computation of the optimal recurrence constant
for a linearly recurrent $k$-automatic sequence.
\end{abstract}

\section{Introduction}

Let ${\bf a} = (a(n))_{n \geq 0}$ be an infinite sequence (or infinite
word) over a finite
alphabet $\Delta$.  We write ${\bf a}[i] = a(i)$, and for
$i, n \geq 0$ we let
${\bf a}[i..i+n-1]$ denote the factor of length $n$ beginning at position $i$.

If a finite word $w$ is expressed in the form $x^n x'$, where
$n \geq 1$ and $x'$ is a prefix of $x$, then we say that $w$ has
{\it period} $x$ and {\it exponent $|w|/|x|$}.  The shortest such period
is called {\it the} period and the largest such exponent is called
{\it the} exponent and is denoted $\exp(w)$.
For example, the period of {\tt alfalfa} is ${\tt alf}$ and 
$\exp({\tt alfalfa}) = 7/3$.
The {\it critical exponent} of an infinite
word $\bf a$ is defined to be the supremum,
over all nonempty factors $w$ of $\bf a$, of the exponent of $w$; it is
denoted by $c({\bf a})$.
It is possible for the critical exponent $c({\bf a})$
to be rational, irrational, or infinite.
If it is rational, it is possible for $c({\bf a})$
to be attained by some finite factor of $\bf a$, or not attained by
any finite factor.

Critical exponents are an active subject of study.  Here are just
a few examples.

\begin{example}
Consider the Thue-Morse sequence
$$ {\bf t} = {\tt 0110100110010110} \cdots ,$$
where ${\bf t}[i]$ is the sum, modulo $2$, of the digits in the
binary expansion of $i$.  Alternatively, $\bf t$ is the fixed point, 
starting with ${\bf 0}$, of the morphism $\mu$ defined by
${\tt 0} \rightarrow {\tt 01}$ and
${\tt 1} \rightarrow {\tt 10}$.

As is well-known, $\bf t$ contains no
overlaps, that is, no factors of the form $axaxa$, where $a \in
\lbrace {\tt 0,1} \rbrace$ and $x \in \lbrace {\tt 0,1} \rbrace^*$.
On the other hand, $\bf t$ contains square factors such as $\tt 00$.
It follows that the critical exponent of $\bf t$ is $2$, and this
exponent is attained by a factor of $\bf t$.
\end{example}

\begin{example}
The sequence ${\tt 0000} \cdots$ clearly has a critical exponent of
$\infty$, as does any ultimately periodic word.
\end{example}

\begin{example}
The Rudin-Shapiro sequence 
${\bf r} = (r_n)_{n \geq 0} = {\tt 0001001000011101} \cdots$
counts the number of (possibly overlapping) occurrences of ${\tt 11}$,
modulo $2$, in the base-$2$ expansion of $n$.  Its critical
exponent is $4$ and it is attained by, for example, the factor
{\tt 0000} beginning at position $7$; see
\cite{Allouche&Bousquet-Melou:1994}.
\label{rudin}
\end{example}

\begin{example}
The sequence ${\bf c} = {\tt 2102012101202102012021012102012} \cdots$, which
counts the number of $\tt 1$'s between consecutive occurrences 
of $\tt 0$ in $\bf t$, is well-known to be squarefree.  However, since
$\bf t$ contains arbitrarily large squares --- for example, the squares
$\mu^n({\tt 00})$ --- it follows that $\bf c$
contains factors of exponent arbitrarily close to $2$.  Thus its
critical exponent is $2$, but this is not attained by any finite factor.
\end{example}

\begin{example}
Consider the Fibonacci word
$${\bf f} = {\tt 0100101001001010010100100101001001} \cdots , $$
defined to be the fixed point of the morphism ${\tt 0} \rightarrow {\tt 01}$
and ${\tt 1} \rightarrow {\tt 0}$.
Then Mignosi and Pirillo \cite{Mignosi&Pirillo:1992}
proved that the
critical exponent of $\bf f$ is $(5+\sqrt{5})/2$, an irrational
number.  
\end{example}

\begin{example}
In fact, every real number greater than $1$ is the critical exponent
of some infinite word \cite{Krieger&Shallit:2007}, and every
real number $\geq 2$ is the critical exponent of some
infinite binary word \cite{Currie&Rampersad:2008}.
\end{example}

Krieger \cite{Krieger:2007,Krieger:2008,Krieger:2009} showed (among other
things) that if an infinite sequence is given as the fixed point of a uniform
morphism, then its critical exponent is either infinite or a rational
number.  

In this paper we generalize this result to the case of $k$-automatic
sequences.  An infinite
sequence $\bf a$ is
said to be {\it $k$-automatic} for some integer $k \geq 2$
if it is computable by a finite automaton
taking as input the base-$k$ representation of $n$, and
having ${\bf a}[n]$ as the output associated with the last state encountered;
see, for example, \cite{Allouche&Shallit:2003b,Cobham:1972}.  

For example, in Figure~\ref{fig1}, we see an automaton generating the
Thue-Morse sequence ${\bf t} = t_0 t_1 t_2 \cdots = {\tt 011010011001}
\cdots$.  The input is $n$, expressed in base $2$, and the output is
the number contained in the state last reached.

\begin{figure}[htb]
\leavevmode
\centerline{\epsfig{file=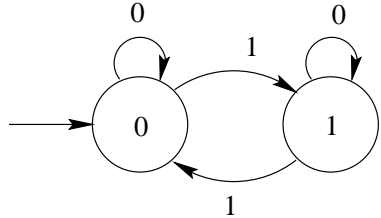,width=5cm}}
\vspace*{8pt}
\protect\label{fig1}
\caption{A finite automaton generating a sequence}
\end{figure}

As is well-known, the class of $k$-automatic sequences
is slightly more general than the class of fixed points of uniform morphisms;
the former also includes words that can be written as the
image, under a coding, of fixed points of uniform morphisms
\cite{Cobham:1972}.  An example of a word that is $2$-automatic but
not the fixed point of any uniform morphism is the Rudin-Shapiro sequence 
$\bf r$,
discussed above in Example~\ref{rudin}.    

(Since this fact does not seem to have
been explicitly proved before, we sketch the proof.  We know that $\bf r$ is
$2$-automatic.  If $\bf r$ were
the fixed point of a $k$-uniform morphism for some $k$ not a power of $2$, then
by Cobham's celebrated theorem \cite{Cobham:1969},
$\bf r$ would be ultimately periodic,
which it is not (since its critical exponent is $4$).  So it must be the
fixed point of a morphism $h$ that is
$2^k$-uniform for some $k \geq 1$.
Now $\bf r$ starts ${\tt 00}$; if ${\bf r} = h({\bf r})$ then
$\bf r$ starts with $h({\tt 0}) h({\tt 0})$.  This means
${\bf r}[2^k - 1] = {\bf r}[2^{k+1} - 1]$.  But clearly the number of
occurrences of ${\tt 11}$ in $2^k - 1$ is one less than the number of
occurrence of ${\tt 11}$ in $2^{k+1} - 1$, a contradiction.)

Allouche, Rampersad, and Shallit
\cite{Allouche&Rampersad&Shallit:2009} proved that the
question

\medskip

\centerline{Given a rational number $r > 1$, is $\bf a$ $r$-power-free?}

\medskip

\noindent is recursively solvable for $k$-automatic sequences
$\bf a$.  More recently,
Charlier, Rampersad, and Shallit \cite{Charlier&Rampersad&Shallit:2011}
showed that 

\medskip

\centerline{Given $\bf a$, is its critical exponent infinite?}

\medskip

\noindent also has a recursive solution for $k$-automatic sequences.

In this paper we show, generalizing some of
the results of Krieger mentioned above,
that the critical exponent of a $k$-automatic sequence is always
either rational or infinite.  Furthermore, we show that
the question

\medskip

\centerline{Given $\bf a$, what is its critical exponent?}

\medskip

\noindent is recursively computable for $k$-automatic sequences.

There are a number of variants of the critical exponent for infinite
words $\bf a$.  For
example, instead of taking the supremum of $\exp(w)$ over all factors
$w$ of $\bf a$, we could take it over only those factors that occur infinitely
often.  Or, letting $x^\beta$ for real $\beta \geq 1$ denote
the shortest prefix of $x^\omega$ of length $\geq \beta|x|$,
we could take the supremum over all real numbers
$\beta$ such that there are arbitrarily large factors of $\bf a$ of the
form $x^\beta$.  We could also restrict our attention to prefixes instead
of factors.  It turns out that for all of these variants,
the resulting critical exponent of automatic sequences is either
rational or infinite, and is computable.

A preliminary version of this paper was presented at the
WORDS 2011 conference in Prague, Czech Republic
\cite{Shallit:2011}.

\section{Two-dimensional automata}

In this paper, we always assume that numbers are encoded in base $k$ using
the digits in $\Sigma_k = \lbrace 0, 1, \ldots, k-1 \rbrace$.

The {\it canonical encoding} of $n$ is the one with no leading zeroes,
and is denoted $(n)_k$.   Thus, for example, we have $(43)_2 = {\tt 101011}$.
Similarly, if $w$ is a word over $\Sigma_k$, then $[w]_k$ denotes the
integer represented by $w$ in base $k$.  Thus $[{\tt 101011}]_2 = 43$.

We will need to encode pairs of
integers.  We handle these
by first padding the representation of
the smaller integer with leading zeroes, so it has the same length
as the larger one,
and then coding the pair as a word over $\Sigma_k^2$.  This gives
the {\it canonical encoding} of a pair $(m,n)$, and is denoted
$(m,n)_k$.  Note that the canonical encoding of a pair does not begin
with a symbol that has $0$ in both components.
For example, the canonical representation of the
pair $(20,13)$ in base $2$ is
$$ [{\tt 1,0}][{\tt 0,1}][{\tt 1,1}][{\tt 0,0}][{\tt 0,1}] ,$$
where the first components spell out ${\tt 10100}$ and the second components
spell out ${\tt 01101}$.   

Given a finite word $x \in (\Sigma_k^2)^*$, we define the {\it projections}
$\pi_i (x)$ ($i = 1,2$) onto the $i$'th coordinate.  
Given a finite word $x$ with $[\pi_2(x)]_k \not= 0$, we define
$$\quo_k(x) = {{[\pi_1(x)]_k} \over {[\pi_2(x)]_k}} .$$  
Thus $\quo_k(x)$ maps words of $(\Sigma_k^2)^*$ to the
non-negative rational numbers $\Que^{\geq 0}$.  (We assume, without
loss of generality, that no denominator is $0$.)
For example, $\quo_2([1,0][1,1][0,1]) = 6/3 = 2$.
If $L \subseteq (\Sigma_k^2)^*$, we define $\quo_k (L) = 
\lbrace \quo_k (x) \ : \ x \in L \rbrace$.

Usually we will assume that the base-$k$ representation is given 
with the most significant digit first, but sometimes, as in the following
result, it is easier to deal with the reversed representations, where the
least significant digit appears first (and shorter representations, if
necessary, are padded with trailing zeroes).  Since the class of regular
languages is (effectively) closed under the map $L \rightarrow L^R$ that
sends a regular language to its reversal, this distinction is not crucial
to our results, and we will not emphasize it unduly.

\begin{lemma}
Let $\beta$ be a non-negative real number and define the
languages
$$L_{\leq \beta} = \lbrace x \in (\Sigma_k^2)^* \ : \ \quo_k(x) \leq \beta \rbrace,$$
and analogously for the relations $<, =, \geq, >, \not=$.
\begin{itemize}
\item[(a)] If $\beta$ is a rational number, then the language
$L_{\leq \beta}$ (resp., $L_{<\beta}$, $L_{=\beta}$, $L_{\geq \beta}$,
$L_{>\beta}$, $L_{\not=\beta}$) is regular.

\item[(b)] If $L_{\leq \beta}$ (resp., $L_{<\beta}$, 
$L_{\geq \beta}$, $L_{>\beta}$) is regular,
then $\beta$ is a rational number.
\end{itemize}
\label{lem1}
\end{lemma}

\begin{proof}  
We handle only the case $L_{\leq \beta}$, the others being similar.

Suppose $\beta$ is rational.  Then we can write $\beta = P/Q$ for integers
$P \geq 0$, $Q \geq 1$.  On input $x$ representing a pair of integers
$(p,q)$ in the reversed base-$k$ representation, we need to accept
iff $p/q \leq P/Q$, that is, iff $pQ \leq qP$.  To do so, we simply
transduce $p$ and $q$ to $pQ$ and $qP$ on the fly, respectively,
and compare them digit-by-digit.  Minor complications arise if the base-$k$
expansions of $pQ$ and $qP$ have different numbers of digits.  To handle
this, we accept if some path ending in $[0,0]^i$ leads to an accepting
condition.  This construction was already given in
\cite{Allouche&Rampersad&Shallit:2009}, and
the details can be found there.

For the other direction, we use ordinary (most-significant-digit first)
representation.  Without loss of generality, we can assume
$1/k \leq \beta < 1$; if not, we can ensure this condition holds
by modifying the automaton, shifting one
coordinate to the left or right.
Take $L_{\leq \beta}$ and intersect with the (regular)
language of words whose second coordinates are of the form $10^*$;
then project
onto the first coordinate to get $L'$, a regular language over
$\Sigma_k$.  Now take the lexicographically
largest word of each length in $L'$ to get $L''$; by a well-known result 
(e.g., \cite{Shallit:1994}), this language
is also regular.  But $L''$ has exactly one word of each length, so
by another well-known result (e.g., \cite{Kunze&Shyr&Thierrin:1981,Paun&Salomaa:1995,Shallit:1994}),
$L''$ must be a finite union of languages
of the form $u v^* w$.
But then $\beta$ is rational, as it is given
by a number whose base-$k$ representation is
$.uvvv \cdots$ for some words $u, v$.
\end{proof}

\section{Computing the $\sup$}
\label{sup}

In this section, we show that if $L \subseteq (\Sigma_k^2)^*$
is a regular language, then
$\alpha := \sup \, \quo_k (L)$
is either rational or infinite, and in both cases it
is computable.

First, we handle the case where the sup is infinite.

\begin{theorem}
Let $L \subseteq (\Sigma_k^2)^*$ be a language accepted by a DFA
with $n$ states.  Assume that no word of $L$ contains leading 
$0$'s.  Then $\sup \, \quo_k (L) = \infty$ if and only if 
$\quo_k (L) \ \cap \ I[k^n, \infty)$ is nonempty.
\label{infthm}
\end{theorem}

\begin{proof}
If $\sup\, \quo_k (L) = \infty$,
then clearly $\quo_k (L) \ \cap \ I[k^n, \infty)$ is nonempty.

For the other direction, suppose $\quo_k (L) \ \cap \ I[k^n, \infty)$
is nonempty.  Then $(p,q)_k \in L$ for some integers $p, q$ with
$p \geq k^n q$.   Writing $x = (p,q)_k$, we have $|x| \geq n$, so we
can apply the pumping lemma, writing $x = uvw$ with $|uv| \leq n$ and
$|v| \geq 1$, and $u v^i w \in L$ for all $i \geq 0$.  
Since $p \geq k^n q$, we must have $[\pi_2(uv)]_k = 0$.    Since
$x$ doesn't start with $[0,0]$, we have $[\pi_1(uv^i w)]_k  \rightarrow \infty$.
Hence $\quo_k(u v^i w) \rightarrow \infty$.
\end{proof}

\begin{corollary}
There is an algorithm that,
given a DFA $M$ accepting a regular language $L \subseteq (\Sigma_k^2)^*$,
decides if $\sup\, \quo_k (L(M)) = \infty$.  
\label{infcor}
\end{corollary}

\begin{proof}
We can find a DFA $M'$ accepting $L$ with all leading $0$'s removed
from words.  If $M'$ has $n$ states, then we can compute a DFA $M''$
accepting $L(M') \ \cap \ L_{>k^n}$, and we can decide if $M''$
accepts anything.
\end{proof}

Next, we turn to the case where the sup is finite.
We start with two useful lemmas.  The first is the classical mediant
inequality.

\begin{lemma}
Let $a, b, c, d$ be non-negative real numbers with $c, d \not= 0$ and
${a \over c} < {b \over d}$.  Then
${a \over c} < {{a+b} \over {c+d}} < {b \over d}$. 
\label{mediant}
\end{lemma}

The next is a fundamental inequality for $\quo$.

\begin{lemma}
Let $u, v, w \in (\Sigma_k^2)^*$ such that
$|v| \geq 1$, and such that
$[\pi_1(uvw)]_k$ and
$[\pi_2(uvw)]_k$ are not both $0$.  
Define 
\begin{equation}
\gamma(u,v) := {{[\pi_1(uv)]_k - [\pi_1(u)]_k} \over
	{[\pi_2(uv)]_k - [\pi_2(u)]_k}}
\label{gam}
\end{equation}
and
\begin{equation}
U := \begin{cases}
\quo_k(w), & \text{if $[\pi_1(uv)]_k = [\pi_2(uv)]_k = 0$}; \\
\infty, & \text{if $[\pi_1(uv)]_k > 0$ and $[\pi_2(uv)]_k = 0$}; \\
\gamma(u,v),	& \text{otherwise.}
\end{cases}
\end{equation}
Then exactly one of the following cases occurs:
\begin{itemize}
\item[(i)] $\quo_k(uw) < \quo_k(uvw) < \quo_k(uv^2w) < \cdots < U$ ;
\item[(ii)] $\quo_k(uw) = \quo_k(uvw) = \quo_k(uv^2w) = \cdots = U$ ;
\item[(iii)] $\quo_k(uw) > \quo_k(uvw) > \quo_k(uv^2w) > \cdots > U$ .
\end{itemize}
Furthermore, $\lim_{i \rightarrow \infty} \quo_k(uv^iw) = U$.
\label{mainl}
\end{lemma}

\begin{proof}
Fix an integer $i \geq 0$.  Define
$$ A_j := [\pi_j(uv)]_k - [\pi_j(u)]_k$$
and
$$ B_j := [\pi_j(uv^iw)]_k$$
for $j = 1,2$
and define $C := k^{i|v|+|w|}$.
Then 
$$ [\pi_j(uv^{i+1}w)]_k = A_j C + B_j $$
for $j = 1,2$.
It follows that 
\begin{eqnarray}
 \quo_k(uv^{i+1}w) - \quo_k(uv^iw) &=&
	{{[\pi_1(uv^{i+1}w)]_k}\over
	{[\pi_2(uv^{i+1}w)]_k}} -
	{{[\pi_1(uv^iw)]_k}\over
	{[\pi_2(uv^iw)]_k}}   \nonumber \\
&=& {{A_1C + B_1} \over {A_2C+B_2}}  - {{B_1} \over {B_2}}. 
\label{quok}
\end{eqnarray}
From the mediant inequality (Lemma~\ref{mediant}) we have
$${{B_1} \over {B_2}} \rela {{A_1} \over {A_2}}
 \quad \implies \quad
 {{B_1} \over {B_2}} \rela {{A_1C + B_1} \over {A_2 C + B_2}}
	\rela {{A_1} \over {A_2}} $$
where $\rela$ is any one of the three relations
$<, =, > $. 
In other words,
$$ \quo_k(uv^i w) \rela U \quad \implies \quad
	\quo_k(uv^iw ) \rela \quo_k(uv^{i+1}w) \rela U.$$
Take $i = 0$ and apply induction to get 
\begin{eqnarray*}
\exists i \quad \quo_k(uv^iw) < U \quad &\implies & \quad \text{ case (i) holds} \\
\exists i \quad \quo_k(uv^iw) = U \quad &\implies & \quad \text{ case (ii) holds } \\
\exists i \quad \quo_k(uv^iw) > U \quad &\implies & \quad \text{ case (iii) holds}. \\
\end{eqnarray*}
This proves our first assertion.

We now prove the assertion about the limit.
Let $j \in \lbrace 1,2\rbrace$, let $i$ be an integer $\geq 1$, and
consider the base-$k$ representation of the rational number
$${{[\pi_j(uv^iw)]_k} \over {k^{i|v|+|w|}}};$$
it looks like 
$$\pi_j(u).\overbrace{\pi_j(v)\pi_j(v) \cdots \pi_j(v)}^i \pi_j (w).$$
On the other hand, the base-$k$ representation of
$$ [\pi_j(u)]_k + {{[\pi_j(v)]_k} \over {k^{|v|} - 1}} $$
looks like
$$ \pi_j(u).\pi_j(v) \pi_j(v) \cdots .$$
Subtracting, we get
$$ \left| {{[\pi_j(uv^iw)]_k} \over {k^{i|v|+|w|}}} -
\left( [\pi_j(u)]_k + {{[\pi_j(v)]_k} \over {k^{|v|} - 1}} \right) \right|
< k^{-i|v|}.$$
It follows that 
$$\lim_{i \rightarrow \infty} {{[\pi_j(uv^iw)]_k} \over {k^{i|v|+|w|}}}
= [\pi_j(u)]_k + {{[\pi_j(v)]_k} \over {k^{|v|} - 1}}.$$
Furthermore, this limit is $0$ if and only if
$[\pi_j(uv)]_k = 0 $.
Hence, provided $[\pi_2(uv)]_k \not= 0$, we get
\begin{eqnarray*}
\lim_{i \rightarrow \infty} \quo_k (u v^i w) &=&
\lim_{i \rightarrow \infty}
{{ [\pi_1(uv^i w)]_k} \over{[\pi_2(uv^i w)]_k}} \\
&=&
\lim_{i \rightarrow \infty}
{{{ [\pi_1(uv^i w)]_k} \over {k^{i|v|+|w|}}} \over
{{ [\pi_2(uv^i w)]_k} \over {k^{i|v|+|w|}}}} \\
&=& {{\lim_{i \rightarrow \infty} {{ [\pi_1(uv^i w)]_k} \over {k^{i|v|+|w|}}} }
\over
{\lim_{i \rightarrow \infty} {{ [\pi_2(uv^i w)]_k} \over {k^{i|v|+|w|}}}}}
\\
&=& {{[\pi_1(u)]_k + {{[\pi_1(v)]_k} \over {k^{|v|} - 1}}} \over
{[\pi_2(u)]_k + {{[\pi_2(v)]_k} \over {k^{|v|} - 1}}}} \\
&=& {{[\pi_1(uv)]_k - [\pi_1(u)]_k} \over
	{[\pi_2(uv)]_k - [\pi_2(u)]_k}} \\
&=& {{A_1}\over{A_2}}.
\end{eqnarray*}

\end{proof}

\begin{theorem}
Let $L \subseteq (\Sigma_k^2)^*$ be a regular language.  Then
$\alpha := \sup\, \quo_k(L)$ is either infinite or rational.
\label{supthm}
\end{theorem}

\begin{proof}
Assume that $\alpha < \infty$.
We will show that $\alpha$ is rational.
In fact, we will show something more:
suppose the DFA $M$ has $n$ states.
Then we claim that $\alpha \in S$, where
$$
S = S_1 \ \cup \ S_2 
$$
and
\begin{eqnarray}
S_1 &=& \lbrace \quo_k (x) \ : \ |x| < n \text{ and } x \in L \rbrace; 
\label{s1} \\
S_2 &=& \lbrace \gamma(u,v) \ : \ |uv| \leq n, \ |v| \geq 1, \text{ and 
there exists $w$ such that $uvw \in L$} \rbrace, \label{s2}
\end{eqnarray}
and $\gamma$ is the function defined in (\ref{gam}).

We will assume, without loss of generality, that no word of $L$ begins
with $[0,0]$.  There are two cases to consider:

\medskip

\noindent Case 1: $\alpha = \quo_k(x)$ for some $x \in L$.
Without loss of generality we can
assume that $x$ is a shortest word achieving the $\sup$. 
We now show $|x| < n$.  If
$|x| \geq n$, then, using the pumping lemma for regular
languages, we can write $x = uvw$ with $|uv| \leq n$ and $|v| \geq 1$,
such that $u v^i w \in L$ for all $i \geq 0$.    Then
by Lemma~\ref{mainl} one of the following two cases must occur:
\begin{itemize}
\item[(a)] $\quo_k(uvw) < \quo_k(uv^2w) < \cdots$;
\item[(b)] $\quo_k(uw) \geq \quo_k(uvw)$;
\end{itemize}
In case (b), we find a shorter word (namely, $uw$),
for which $\quo_k(uw) \geq \quo_k(x)$, contradicting our assumption
that $x$ was the shortest word achieving the $\sup$.  In case (a) we
find a word (namely, $uv^2w$) such that $\quo_k(uv^2w) > \quo_k(x)$,
contradicting the fact that $\quo_k(x) = \sup \, \quo_k (L)$.
Thus $|x| < n$ and hence $\quo_k(x) \in S_1$.

\medskip

\noindent Case 2:  The sup is not achieved on $L$.  Then there must
be an infinite sequence of distinct
words $(x_j)_{j \geq 1}$ with each $x_j \in L$
and $\quo_k(x_j)$ converging strictly monotonically to $\alpha$ from below.
Without loss of generality, we can assume that each $x_j$ is of length
$\geq n$ (the number of states of $M$) and further that 
\begin{equation}
\quo_k(x_j) \geq \quo_k(y) \text{ for each } y \text{ with } 
|y| \leq |x_j|.
\label{eqa}
\end{equation}
By the pumping lemma, we can write each $x_j = u_j v_j w_j$ with
$|u_j v_j| \leq n$ and $|v_j| \geq 1$ such that $u_j v_j^i w_j \in L$
for all $i \geq 0$.

Since $|u_j v_j| \leq n$, there are only finitely
many choices for $u_j v_j$.  By the infinite pigeonhole principle,
there is a {\it single\/} choice of $u_j, v_j$ (say $u,v$) corresponding to
infinitely many decompositions of the $x_j$.  Let us restrict ourselves
to this particular subsequence, which we write as $(x'_j)_{j \geq 1}$.
Thus $x'_j = uvw_j$ for $j \geq 1$.  

Now, appealing once more to Lemma~\ref{mainl}, we see that there
are two possibilities:
\begin{itemize}
\item[(a)] there exists $j$ such that $\quo_k(u w_j) \geq \quo_k(u v w_j)$;
\item[(b)] for all $j \geq 1$ we have
$\quo_k(u w_j) < \quo_k (u v w_j) < \quo_k(uv^2 w_j) < \cdots$.
\end{itemize}
In case (a) we have found a shorter word with a quotient at least as large,
contradicting our assumption (\ref{eqa}).  Hence case (b) must occur.

Since $u v^i w_1 \in L$ for all $j \geq 1$, we have
$\quo_k(u v^i w_1) \leq \alpha$ for all $i \geq 1$, and hence
$\sup_{i \geq 1} \quo_k (u v^i w_1) \leq \alpha$.  
On the other hand,
$\sup_{i \geq 1} \quo_k (u v^i w_1) = \lim_{i \geq 1}
\quo_k (u v^i w_1) = \gamma(u,v)$ by Lemma~\ref{mainl}.
It follows that $\gamma(u,v) \leq \alpha$.

However, for all $j \geq 1$ we have
$\quo_k(u v w_j) \leq \lim_{i \rightarrow \infty} 
\quo_k(u v^i w_j) = \gamma(u,v)$.
Hence
$\sup_{j \geq 1} \quo_k (u v w_j) \leq \gamma(u,v)$.
But
$\alpha = \sup_{j \geq 1} x'_j = \sup_{j \geq 1} \quo_k(u v w_j) 
\leq \gamma(u,v)$.

Putting these two results together, we see that
$\gamma(u,v) = \alpha$, and hence $\alpha \in S_2$. 
\hphantom{aaa} 
\end{proof}

\begin{corollary}
There is an algorithm that,
given a DFA $M$ accepting $L \subseteq (\Sigma_k^2)^*$, 
will compute $\alpha = \sup\, \quo_k(L)$.
\label{computethm}
\end{corollary}

\begin{proof}
Using Corollary~\ref{infcor}, we have
an algorithm to decide if $\alpha$ is infinite.

Otherwise, we know from the proof of Theorem~\ref{supthm} that
$\alpha$ lies in $S_1 \ \cup\  S_2$, where $S_1$ and $S_2$ are
finite sets that we can compute explicitly from $M$.
Furthermore,
$$ \alpha = \min \ \lbrace \beta \in S_1 \ \cup \ S_2 \ : \
	L(M) \ \cap \ L_{>\beta} = \emptyset \rbrace.$$
So it suffices to check,
for each $\beta \in S_1 \cup S_2$, if
the language
$L(M) \ \cap \ L_{>\beta}$ is empty, which can be done using the
usual depth-first search techniques on the automaton for the
intersection. 
\end{proof}

\section{Computing the largest special point}
\label{special}

Let $L \subseteq (\Sigma_k^2)^*$.
We say an extended real number
$\beta$ is a {\it special point} of $\quo_k (L)$ if 
there exists an infinite sequence $(x_j)_{j \geq 1}$
of distinct words of $L$ such that
$\lim_{j \rightarrow \infty} \quo_k (x_j) = \beta$.
Thus a special point is either
an accumulation point of $\quo_k (L)$, or a rational number with infinitely
many distinct representations in $L$.  Note that every infinite language
$L$ has a special point, and indeed, a largest special point.

\begin{theorem}
Let $L$ be an infinite regular language accepted by a DFA with $n$ states.
Then the largest special point of $L$ is either infinite or rational.
\label{specialthm}
\end{theorem}

Before we begin the proof, we need the following lemma.

\begin{lemma}
Let $u, v$ be fixed words such that
$[\pi_2(uv)]_k \not= 0$ and let $i$ be a fixed integer.
Let $(w_j)_{j \geq 1}$ be a sequence of words.
If $\lim_{j \rightarrow \infty} \quo_k ( u v^i w_j) =
\lim_{j \rightarrow \infty} \quo_k (u v^{i+1} w_j)$,
then these limits both equal $\gamma(u,v)$, where $\gamma$ is
defined in (\ref{gam}).
\label{limlem}
\end{lemma}

\begin{proof}
As in the proof of Lemma~\ref{mainl}, define
$A_l = [\pi_l(uv)]_k - [\pi_l(u)]_k$ for $l = 1,2$.
Also define
$$D_{l,j} = [\pi_l ( u v^i w_j )]_k k^{-(i|v|+ |w_j|)}$$
for $l = 1,2$ and $j \geq 1$.  
Then, using (\ref{quok}), the hypothesis on the limits can be restated as
\begin{equation}
\forall \epsilon > 0 \ \exists N \ \forall j \geq N \ 
\left| {{A_1 + D_{1,j}} \over {A_2 + D_{2,j}}} 
- {{D_{1,j}} \over {D_{2,j}}} \right| < \epsilon .
\label{e1}
\end{equation}
Clearing the denominators and simplifying, we see that
(\ref{e1}) implies
\begin{equation}
\forall \epsilon > 0 \ \exists N \ \forall j \geq N \ 
| A_1 D_{2,j} - A_2 D_{1,j} | < \epsilon (A_2 + D_{2,j}) D_{2,j}.
\label{e2}
\end{equation}
Dividing by $A_2 D_{2,j}$, we see that (\ref{e2}) implies
\begin{equation}
\forall \epsilon > 0 \ \exists N \ \forall j \geq N \
\left| {{A_1} \over{A_2}} - {{D_{1,j}} \over {D_{2,j}}} \right|
< \epsilon \left( 1 + {{D_{2,j}} \over {A_2}} \right) .
\label{e3}
\end{equation}
From the hypothesis on $u, v$ we have $A_2 \not= 0$.
But ${{A_1} \over {A_2}} = \gamma(u,v)$ and
${{D_{1,j}} \over {D_{2,j}}} = \quo_k ( u v^i w_j)$,
so (\ref{e3}) can be restated as
$\lim_{j \rightarrow \infty} \quo_k (u v^i w_j ) = \gamma(u,v)$.
\end{proof}

Now we can return to the proof of Theorem~\ref{specialthm}.

\begin{proof}
Let $\alpha$ be the largest special point in $\quo_k (L)$.
Then there is an infinite sequence $(x_j)_{j \geq 1}$ of distinct
words of $L$ such that $\lim_{j \rightarrow \infty} \quo_k (x_j) = \alpha$.
We show that if $\alpha <
\infty$ then $\alpha \in S_2$, where $S_2$ is the set of rationals
defined in (\ref{s2}).

Our proof involves considering more and more refined subsequences
of the $(x_j)$; by abuse of notation we refer to each of these
subsequences as $(x_j)$.

First, we can assume without loss of generality that
$$ n \leq |x_1| < |x_2| < \cdots, $$
where $n$ is the number of states in the minimal DFA accepting $L$.
Now apply the pumping lemma to each $x_j$, obtaining
the decompositions $x_j = u_j v_j w_j$ such that
$|u_j v_j| \leq n$ and $|v_j| \geq 1$ and
$u_j v_j^i w_j \in L$ for all $i \geq 0$.   By the infinite
pigeonhole principle, there must be some $u_j v_j$ that occurs
infinitely often, so by replacing the $(x_j)$ with the appropriate
subsequence, we can also assume that 
the pumping lemma in fact gives the decomposition  $x_j = u v w_j$ 
for each $j \geq 1$.

Applying Lemma~\ref{mainl}, we see that
$\lim_{i \rightarrow \infty} \quo_k(uv^i w_j) = \gamma(u,v)$;
and further, for each $j \geq 1$ we have either
\begin{itemize}
\item[(a)] $\quo_k (uw_j) < \quo_k (uvw_j) < \quo_k (uv^2 w_j) < \cdots
	< \gamma(u,v)$; or 
\item[(b)] $\quo_k (uw_j) \geq \quo_k (uvw_j) \geq \quo_k (uv^2w_j) \geq
\cdots \geq \gamma(u,v)$ .
\end{itemize}
Again, by the infinite pigeonhole principle, at least one of the two
options above must occur for infinitely many
$j$, so by restricting to the
appropriate subsequence,
we can assume that one of the two sets of
inequalities applies for all $j$.    We consider both cases in turn.

\medskip

Case (a):  The sequence $s = (\quo_k(uv^2w_j))_{j \geq 1}$
cannot be unbounded since
$\alpha < \infty$. From the Bolzano-Weierstrass theorem, we know
$s$ has a convergent subsequence, so we can replace $(w_j)$ with the
appropriate subsequence
and define $\beta := \lim_{j \rightarrow \infty} \quo_k(uv^2 w_j)$.
Then $\quo_k(uvw_j) < \quo_k(uv^2w_j)$, so 
\begin{displaymath}
\alpha = \lim_{j \rightarrow \infty} \quo_k(uvw_j) \leq \lim_{j \rightarrow
\infty} \quo_k(uv^2w_j) = \beta.
\end{displaymath}
On the other hand,
$\beta$ is a special point, so $\beta \leq \alpha$. Therefore $\alpha =
\beta$ and Lemma~\ref{limlem} applies,
giving $\alpha = \beta = \gamma(u, v)$.

\medskip

Case (b):  Just like Case (a), except now we consider the
sequence $s = (\quo_k(u w_j))_{j \geq 1}$ instead.

\medskip

In both cases, then, we have shown that $\alpha = \gamma(u,v) \in S_2$, and so
$\alpha$ is rational.
\hphantom{aaa} 
\end{proof}

\begin{corollary}
There is an algorithm that,
given a DFA $M$ accepting an infinite language $L \subseteq (\Sigma_k^2)^*$,
will compute the largest special point in $\quo_k (L)$.
\label{specialcor}
\end{corollary}

\begin{proof}
In Theorem~\ref{specialthm} we showed that the largest special point
is either $\infty$ or contained in the set
$$S_2 = \lbrace \gamma(u,v) \ : \ |uv| \leq n, \ |v| \geq 1, \text{ and 
there exists $w$ such that $uvw \in L$} \rbrace .$$
The former case occurs iff
$\sup\, \quo_k (L) = \infty$,
which can be checked using Corollary~\ref{infcor}.

Otherwise, we (effectively) can list the (finite number of) elements
of $S_2$.  For each $\beta \in S_2$, we can check to see if $\beta$ is
an accumulation point using \cite{Rowland&Shallit:2011}, Thm.\ 24.
We can also check to see if $\beta$ has infinitely many representations in
$L$ by computing a DFA accepting $L \ \cap \ L_{=\beta}$ and then
using the usual method involving cycle detection via depth-first search.
Now $\alpha$ is the largest such $\beta$ for which one of the two
conditions applies.
\end{proof}

Every accumulation point is a special point, but the converse is not
necessarily true.  If, however, our language $L$ has a certain natural
property, then the converse holds.

\begin{theorem}
Let $L \subseteq (\Sigma_k^2)^*$ be a language such that
\begin{itemize}
\item[(a)]  No word of $L$ has leading $[0,0]$'s;
\item[(b)] $\quo_k (L)$ has infinite cardinality;
\item[(c)] If $(p,q)_k \in L$, then $p \geq q$;
\item[(d)] If $(p,q)_k \in L$ and $p > q$ then $(p-1,q)_k \in L$.
\end{itemize}
Then every special point, except perhaps $1$, is an accumulation point.
Furthermore, if (a)--(d) hold, the largest accumulation point
(that is, $\limsup \, \quo_k (L)$) is rational or
infinite, and is computable.
\label{limsupthm}
\end{theorem}

\begin{proof}
Suppose the conclusion is false.  Then there is a special point
$\alpha > 1$ that is not an accumulation point.  Then there
are infinitely many representations of $\alpha$ in $L$;
choose a sequence of these $(x_i)_{i \geq 1}$ of increasing length.
Writing $x_i = (p_i,q_i)_k$ with $p_i > q_i$, by hypothesis
we get $(p_i-1, q_i)_k \in L$.    But evidently
$\lim_{i \rightarrow \infty} (p_i - 1)/q_i = \alpha$, so
$\alpha$ is an accumulation point, a contradiction.

The results on rationality and computability now follow
from Theorem~\ref{specialthm} and Corollary~\ref{specialcor}.
\end{proof}

\section{Application to the critical exponent and its variants}

We can now apply the results of Sections~\ref{sup} and \ref{special}
to the critical
exponent problem.

\begin{theorem}
The critical exponent of a $k$-automatic sequence
is either
rational or infinite, and is effectively computable.
\label{effectivethm}
\end{theorem}

\begin{proof}
Given a $k$-automatic sequence ${\bf a} = (a_i)_{i \geq 0}$, we can,
using the techniques of \cite{Allouche&Rampersad&Shallit:2009,Charlier&Rampersad&Shallit:2011},
(effectively) create a two-dimensional DFA $M$ accepting
\begin{eqnarray*}
L' &=& \lbrace (p,q)_k \ : \ \exists \text{ a factor of $\bf a$ of length
	$q$ with period $p$ } \rbrace  \\
&=& \lbrace (p,q)_k \ : \ \exists i \text{ such that }
	{\bf a}[i..i+q-p-1] = {\bf a}[i+p..i+q-1] \rbrace  \\
&=& \lbrace (p,q)_k \ : \ \exists i \ \forall j, \ 0 \leq j < q-p
	\text{ we have } {\bf a}[i] = {\bf a}[i+p] \rbrace .  \\
\end{eqnarray*}
Then the critical exponent of $\bf a$ is
$\sup\, \quo_k(L')$, which, by Theorem~\ref{supthm}, is rational or
infinite.  The infinite case has already been handled in
\cite{Charlier&Rampersad&Shallit:2011} (or we could use
Corollary~\ref{infcor}).   If it is finite,
Corollary~\ref{computethm}
tells us how to compute it from $M$.
\end{proof}

The same results also hold for the variant of the critical
exponent when the sup is taken over only the factors that occur
infinitely often.

\begin{theorem}
If ${\bf a} = (a_i)_{i \geq 0}$ is a
$k$-automatic sequence, 
the quantity 
$$ c_1({\bf a}) := \sup \lbrace \exp(w) \ : \ 
w \text{ is a finite factor of } {\bf a} \text{ that occurs infinitely
often } \rbrace $$
is either rational or infinite, and is computable.
\end{theorem}

\begin{proof}
To see this, it is only necessary to change the appropriate 
two-dimensional DFA to accept
\begin{multline}
L'' = \lbrace (p,q)_k \ : \ \exists i \text{ such that }
	{\bf a}[i..i+q-p-1] = {\bf a}[i+p..i+q-1]  \\
	\text{ and for all } j \text{ such that }
	{\bf a}[i..i+q-1] = {\bf a}[j..j+q-1] \\
	\text{ there exists } \ell > j \text{ such that }
	{\bf a}[i..i+q-1] = {\bf a}[\ell..\ell+q-1]  \rbrace
\end{multline}
The first clause says that the factor of length $q$ at position
$i$ has period $p$, and the other two clauses say that if some
factor equals this one, then there is another occurrence of that
factor further on.

Now apply Theorem~\ref{supthm} and Corollary~\ref{computethm} to $L''$.
\end{proof}

The same results also hold for the variant where 
we only consider exponents $\beta$ that work for arbitrarily large
factors.  This corresponds precisely to our notion of special
point introduced above, 
so that the largest special point of
$\quo_k (L')$ gives the supremum over all such $\beta$.

\begin{theorem}
If ${\bf a} = (a_i)_{i \geq 0}$ is a
$k$-automatic sequence, the quantity 
$$ c_2({\bf a}) := \sup \lbrace \beta \ : \ 
\forall N \geq 1 \ \exists w \text{ with } |w| \geq N
\text{ and } w^\beta \text{ a factor of } {\bf a} \rbrace$$
is either rational or infinite, and is computable.
\label{spec2}
\end{theorem}

\begin{proof}
Consider $L'$ as defined in the proof of Theorem~\ref{effectivethm}.  Then
$ c_2({\bf a})$ is the largest special point in $\quo_k(L')$.
Now apply Theorem~\ref{specialthm} and Corollary~\ref{specialcor}.
\end{proof}

Yet another variation is the
{\it initial critical exponent} ${\rm ice}_1 ({\bf a})$, introduced in
\cite{Berthe&Holton&Zamboni:2006} (also see \cite{Adamczewski:2010}),
where the supremum of exponents
is taken over all {\it prefixes} of a given infinite word,
as opposed to all factors.   
There is also the variant ${\rm ice}_2$, where we consider only the
exponents that work for {\it arbitrarily large} prefixes.
Our results also apply to both these cases.

\begin{theorem}
Let ${\rm ice}_1 ({\bf a}) = \sup \lbrace \exp(w) \ : \ w
\text{ is a prefix of } {\bf a} \rbrace$ and
let ${\rm ice}_2 ({\bf a}) = \sup \lbrace \beta \ : \ 
\forall N \geq 1 \ \exists w \text{ with } |w| \geq N
\text{ and } w^\beta \text{ a prefix of } {\bf a} \rbrace.$
If $\bf a$ is a $k$-automatic sequence, then both the quantities
${\rm ice}_1 ({\bf a})$ and
${\rm ice}_2 ({\bf a})$ are either rational or infinite,
and are computable.
\end{theorem}

\begin{proof}
Here the proof is exactly like the proofs of 
Theorems~\ref{effectivethm} and \ref{spec2}, with the difference that we
replace $L'$ defined in the proof of Theorem~\ref{effectivethm}
with the language
$$L_i = \lbrace (p,q)_k \ : \ {\bf a}[0..q-p-1] = {\bf a}[p..q-1] \rbrace .$$
Then ${\rm ice}_1 ({\bf a}) = \sup \quo_k (L_i)$,
while ${\rm ice}_2 ({\bf a}) $ is the largest special point in
$\quo_k (L_i)$.  
\hphantom{aaa} 
\end{proof}

Finally, our results also apply to the so-called {\it Diophantine
exponent} ${\rm Dio}({\bf a})$, introduced in \cite{Adamczewski&Bugeaud:2007}
(also see \cite{Adamczewski:2010}).  It is defined as
the supremum of real numbers $\beta$ for which
there exist arbitrarily long prefixes of $\bf a$ 
that can be expressed in the form as $u v^\tau$ for some real number $\tau$
and finite words $u, v$ such that
$|u v^\tau | / |uv| \geq \beta$.    The following results complement
those in \cite{Bugeaud&Krieger&Shallit:2011}.

\begin{theorem}
If $\bf a$ is a $k$-automatic sequence, then
${\rm Dio}({\bf a})$ is either rational or infinite, and is computable.
\end{theorem}

\begin{proof}
Here we follow the proof of Theorem~\ref{spec2} once more,
except now we work with the language
$$L_d = \lbrace (i+\ell,i+p)_k \ : \exists i \geq 0, \ell \geq p \geq 1
\text{ such that } {\bf a}[i..i+\ell-p-1] =
{\bf a}[i+p..i+\ell-1] \rbrace.$$  
The claims now follow from Theorem~\ref{specialthm} and Corollary~\ref{specialcor}.
\end{proof}

\section{Other applications}

Theorem~\ref{supthm} and Corollary~\ref{computethm}
have applications to other problems.

A sequence ${\bf a}$ is said to be {\it recurrent} if every factor that occurs,
occurs infinitely often.  It is {\it linearly recurrent} if there exists
a constant $C$ such that for all $\ell \geq 0$, and all factors $x$
of length $\ell$ occurring in $\bf a$, any two consecutive occurrences of
$x$ are separated by at most $C\ell$ positions.  

\begin{theorem}
It is decidable if a $k$-automatic sequence $\bf a$
is linearly recurrent.  If $\bf a$
is linearly recurrent, the optimal constant $C$ is computable.
\end{theorem}

\begin{proof}
First, as in \cite{Charlier&Rampersad&Shallit:2011}, we construct an
automaton accepting the language
\begin{align*}
L &= \{ (n,l)_k  \ : \\
& \text{ (a) there exists $i \geq 0$ s.\ t.\ 
${\bf a}[i+j] = {\bf a}[i+n+j]$ for all $j,\ 0 \leq j < \ell$, and} \\
& \text{ (b) there is no $t,\ 0 < t < n$ s.\ t.\ 
${\bf a}[i+j] = {\bf a}[i+t+j]$ for all $j, 0 \leq j < \ell$ } \} 
\end{align*}

Another way to say this is that $L$ consists of the base-$k$ representation
of those pairs of integers $(n,\ell)$ such
that (a) there is some factor of length $\ell$ for which
there is another occurrence at distance $n$ and (b) this
occurrence is actually the very next occurrence.

Now from Theorem~\ref{supthm} we know that
$\sup \{ n/\ell \ :  \ (n,\ell)_k \in L \}$
is either infinite or rational.  In the
latter case this sup is computable, by Corollary~\ref{computethm} 
and this gives the optimal constant $C$
for the linear recurrence of $\bf a$.
\end{proof}

\section{Open problems}

In this paper we have examined $\sup_{x \in L} \quo_k(x)$ for $L$ 
regular.
We do not currently know how to prove analogous results
for $L$ context-free.
Nor do we know how to extend the results on critical exponents
to the more general case of morphic sequences.

\section{Acknowledgments}

In a previous version of this paper that was presented at WORDS 2011,
the second author stated that if $L$ is regular,
then $\limsup_{x \in L} \quo_k(x)$ is
either rational or infinite, and is computable \cite{Shallit:2011}.
Unfortunately the envisioned proof of this more general result
contained a subtle flaw that we have not been able to repair.
The present paper contains slightly weaker results.
We are indebted to Micha\"el Cadilhac for raising the question.

We would like to thank Wojciech Rytter, who asked a question at
the 2011 Dagstuhl meeting on combinatorics on words that led to this paper.
We also thank Eric Rowland and Victor Moll, whose hospitality at Tulane
University permitted the writing of the first draft
of this paper.  Finally, we thank
Jean-Paul Allouche, Yann Bugeaud, and \'Emilie Charlier
for their helpful suggestions.


\begin{thebibliography}{10}
\providecommand{\bibitemdeclare}[2]{}
\providecommand{\urlprefix}{Available at }
\providecommand{\url}[1]{\texttt{#1}}
\providecommand{\href}[2]{\texttt{#2}}
\providecommand{\urlalt}[2]{\href{#1}{#2}}
\providecommand{\doi}[1]{doi:\urlalt{http://dx.doi.org/#1}{#1}}
\providecommand{\bibinfo}[2]{#2}

\bibitemdeclare{article}{Adamczewski:2010}
\bibitem{Adamczewski:2010}
\bibinfo{author}{B.~Adamczewski} (\bibinfo{year}{2010}):
  \emph{\bibinfo{title}{On the expansion of some exponential periods in an
  integer base}}.
\newblock {\sl \bibinfo{journal}{Math. Ann.}} \bibinfo{volume}{346}, pp.
  \bibinfo{pages}{107--116}.

\bibitemdeclare{article}{Adamczewski&Bugeaud:2007}
\bibitem{Adamczewski&Bugeaud:2007}
\bibinfo{author}{B.~Adamczewski} \& \bibinfo{author}{Y.~Bugeaud}
  (\bibinfo{year}{2007}): \emph{\bibinfo{title}{Dynamics for $\beta$-shifts and
  {Diophantine} approximation}}.
\newblock {\sl \bibinfo{journal}{Ergodic Theory Dynam. Systems}}
  \bibinfo{volume}{27}, pp. \bibinfo{pages}{1695--1711}.

\bibitemdeclare{article}{Allouche&Bousquet-Melou:1994}
\bibitem{Allouche&Bousquet-Melou:1994}
\bibinfo{author}{J.-P. Allouche} \& \bibinfo{author}{M.~Bousquet-M\'elou}
  (\bibinfo{year}{1994}): \emph{\bibinfo{title}{Facteurs des suites de
  {Rudin-Shapiro} {g\'en\'eralis\'ees}}}.
\newblock {\sl \bibinfo{journal}{Bull. Belg. Math. Soc.}} \bibinfo{volume}{1},
  pp. \bibinfo{pages}{145--164}.

\bibitemdeclare{article}{Allouche&Rampersad&Shallit:2009}
\bibitem{Allouche&Rampersad&Shallit:2009}
\bibinfo{author}{J.-P. Allouche}, \bibinfo{author}{N.~Rampersad} \&
  \bibinfo{author}{J.~Shallit} (\bibinfo{year}{2009}):
  \emph{\bibinfo{title}{Periodicity, repetitions, and orbits of an automatic
  sequence}}.
\newblock {\sl \bibinfo{journal}{Theoret. Comput. Sci.}} \bibinfo{volume}{410},
  pp. \bibinfo{pages}{2795--2803}, \doi{10.1016/j.tcs.2009.02.006}.

\bibitemdeclare{book}{Allouche&Shallit:2003b}
\bibitem{Allouche&Shallit:2003b}
\bibinfo{author}{J.-P. Allouche} \& \bibinfo{author}{J.~Shallit}
  (\bibinfo{year}{2003}): \emph{\bibinfo{title}{Automatic Sequences: Theory,
  Applications, Generalizations}}.
\newblock \bibinfo{publisher}{Cambridge University Press}.

\bibitemdeclare{article}{Berthe&Holton&Zamboni:2006}
\bibitem{Berthe&Holton&Zamboni:2006}
\bibinfo{author}{V.~{Berth\'e}}, \bibinfo{author}{C.~Holton} \&
  \bibinfo{author}{L.~Q. Zamboni} (\bibinfo{year}{2006}):
  \emph{\bibinfo{title}{Initial powers of {Sturmian} sequences}}.
\newblock {\sl \bibinfo{journal}{Acta Arith.}} \bibinfo{volume}{122}, pp.
  \bibinfo{pages}{315--347}.

\bibitemdeclare{article}{Bugeaud&Krieger&Shallit:2011}
\bibitem{Bugeaud&Krieger&Shallit:2011}
\bibinfo{author}{Y.~Bugeaud}, \bibinfo{author}{D.~Krieger} \&
  \bibinfo{author}{J.~Shallit} (\bibinfo{year}{2011}):
  \emph{\bibinfo{title}{Morphic and automatic words: maximal blocks and
  Diophantine approximation}}.
\newblock {\sl \bibinfo{journal}{Acta Arith.}} \bibinfo{volume}{149}, pp.
  \bibinfo{pages}{181--199}, \doi{10.4064/aa149-2-7}.

\bibitemdeclare{incollection}{Charlier&Rampersad&Shallit:2011}
\bibitem{Charlier&Rampersad&Shallit:2011}
\bibinfo{author}{E.~Charlier}, \bibinfo{author}{N.~Rampersad} \&
  \bibinfo{author}{J.~Shallit} (\bibinfo{year}{2011}):
  \emph{\bibinfo{title}{Enumeration and decidable properties of automatic
  sequences}}.
\newblock In \bibinfo{editor}{G.~Mauri} \& \bibinfo{editor}{A.~Leporati},
  editors: {\sl \bibinfo{booktitle}{Developments in Language Theory, 15th
  International Conference, DLT 2011}}, {\sl \bibinfo{series}{Lect. Notes in
  Comput. Sci.}} \bibinfo{volume}{6795}, pp. \bibinfo{pages}{165--179}.
\newblock \bibinfo{note}{Available at {\tt http://arxiv.org/abs/1102.3698}}.

\bibitemdeclare{article}{Cobham:1969}
\bibitem{Cobham:1969}
\bibinfo{author}{A.~Cobham} (\bibinfo{year}{1969}): \emph{\bibinfo{title}{On
  the base-dependence of sets of numbers recognizable by finite automata}}.
\newblock {\sl \bibinfo{journal}{Math. Systems Theory}} \bibinfo{volume}{3},
  pp. \bibinfo{pages}{186--192}, \doi{10.1007/BF01746527}.

\bibitemdeclare{article}{Cobham:1972}
\bibitem{Cobham:1972}
\bibinfo{author}{A.~Cobham} (\bibinfo{year}{1972}):
  \emph{\bibinfo{title}{Uniform tag sequences}}.
\newblock {\sl \bibinfo{journal}{Math. Systems Theory}} \bibinfo{volume}{6},
  pp. \bibinfo{pages}{164--192}, \doi{10.1007/BF01706087}.

\bibitemdeclare{article}{Currie&Rampersad:2008}
\bibitem{Currie&Rampersad:2008}
\bibinfo{author}{J.~D. Currie} \& \bibinfo{author}{N.~Rampersad}
  (\bibinfo{year}{2008}): \emph{\bibinfo{title}{For each $\alpha > 2$ there is
  an infinite binary word with critical exponent $\alpha$}}.
\newblock {\sl \bibinfo{journal}{Elect. J. Combinatorics}}
  \bibinfo{volume}{15}:\bibinfo{eid}{\#N34}.
\newblock \bibinfo{note}{Available at {\tt
  http://www.combinatorics.org/Volume\_15/Abstracts/v15i1n34.htm}}.

\bibitemdeclare{article}{Krieger:2007}
\bibitem{Krieger:2007}
\bibinfo{author}{D.~Krieger} (\bibinfo{year}{2007}): \emph{\bibinfo{title}{On
  critical exponents in fixed points of non-erasing morphisms}}.
\newblock {\sl \bibinfo{journal}{Theor. Comput. Sci.}} \bibinfo{volume}{376},
  pp. \bibinfo{pages}{70--88}, \doi{10.1016/j.tcs.2007.01.020}.

\bibitemdeclare{phdthesis}{Krieger:2008}
\bibitem{Krieger:2008}
\bibinfo{author}{D.~Krieger} (\bibinfo{year}{2008}):
  \emph{\bibinfo{title}{Critical exponents and stabilizers of infinite words}}.
\newblock Ph.D. thesis, \bibinfo{school}{University of Waterloo}.

\bibitemdeclare{article}{Krieger:2009}
\bibitem{Krieger:2009}
\bibinfo{author}{D.~Krieger} (\bibinfo{year}{2009}): \emph{\bibinfo{title}{On
  critical exponents in fixed points of $k$-uniform binary morphisms}}.
\newblock {\sl \bibinfo{journal}{RAIRO Info. Theor. Appl.}}
  \bibinfo{volume}{43}, pp. \bibinfo{pages}{41--68}, \doi{10.1051/ita:2007042}.

\bibitemdeclare{article}{Krieger&Shallit:2007}
\bibitem{Krieger&Shallit:2007}
\bibinfo{author}{D.~Krieger} \& \bibinfo{author}{J.~Shallit}
  (\bibinfo{year}{2007}): \emph{\bibinfo{title}{Every real number greater than
  $1$ is a critical exponent}}.
\newblock {\sl \bibinfo{journal}{Theoret. Comput. Sci.}} \bibinfo{volume}{381},
  pp. \bibinfo{pages}{177--182}, \doi{10.1016/j.tcs.2007.04.037}.

\bibitemdeclare{article}{Kunze&Shyr&Thierrin:1981}
\bibitem{Kunze&Shyr&Thierrin:1981}
\bibinfo{author}{M.~Kunze}, \bibinfo{author}{H.~J. Shyr} \&
  \bibinfo{author}{G.~Thierrin} (\bibinfo{year}{1981}):
  \emph{\bibinfo{title}{$h$-bounded and semidiscrete languages}}.
\newblock {\sl \bibinfo{journal}{Info. Control}} \bibinfo{volume}{51}, pp.
  \bibinfo{pages}{147--187}, \doi{10.1016/S0019-9958(81)90253-9}.

\bibitemdeclare{article}{Mignosi&Pirillo:1992}
\bibitem{Mignosi&Pirillo:1992}
\bibinfo{author}{F.~Mignosi} \& \bibinfo{author}{G.~Pirillo}
  (\bibinfo{year}{1992}): \emph{\bibinfo{title}{Repetitions in the {Fibonacci}
  infinite word}}.
\newblock {\sl \bibinfo{journal}{RAIRO Info. Theor. Appl.}}
  \bibinfo{volume}{26}, pp. \bibinfo{pages}{199--204}.

\bibitemdeclare{article}{Paun&Salomaa:1995}
\bibitem{Paun&Salomaa:1995}
\bibinfo{author}{G.~P\v{a}un} \& \bibinfo{author}{A.~Salomaa}
  (\bibinfo{year}{1995}): \emph{\bibinfo{title}{Thin and slender languages}}.
\newblock {\sl \bibinfo{journal}{Discrete Appl. Math.}} \bibinfo{volume}{61},
  pp. \bibinfo{pages}{257--270}, \doi{10.1016/0166-218X(94)00014-5}.

\bibitemdeclare{unpublished}{Rowland&Shallit:2011}
\bibitem{Rowland&Shallit:2011}
\bibinfo{author}{E.~Rowland} \& \bibinfo{author}{J.~Shallit}
  (\bibinfo{year}{2011}): \emph{\bibinfo{title}{$k$-automatic sets of rational
  numbers}}.
\newblock \bibinfo{note}{To appear, LATA 2011 conference. Available at {\tt
  http://arxiv.org/abs/1110.2382}}.

\bibitemdeclare{article}{Shallit:1994}
\bibitem{Shallit:1994}
\bibinfo{author}{J.~Shallit} (\bibinfo{year}{1994}):
  \emph{\bibinfo{title}{Numeration systems, linear recurrences, and regular
  sets}}.
\newblock {\sl \bibinfo{journal}{Inform. Comput.}} \bibinfo{volume}{113}, pp.
  \bibinfo{pages}{331--347}, \doi{10.1006/inco.1994.1076}.

\bibitemdeclare{incollection}{Shallit:2011}
\bibitem{Shallit:2011}
\bibinfo{author}{J.~Shallit} (\bibinfo{year}{2011}): \emph{\bibinfo{title}{The
  critical exponent is computable for automatic sequences}}.
\newblock In \bibinfo{editor}{P.~Ambroz}, \bibinfo{editor}{S.~Holub} \&
  \bibinfo{editor}{Z.~{Mas\'akov\'a}}, editors: {\sl
  \bibinfo{booktitle}{Proceedings 8th International Conference, WORDS 2011}},
  pp. \bibinfo{pages}{231--239}.
\newblock \bibinfo{note}{Vol.\ 63 of {\it Elect. Proc. Theor. Comput. Sci.}.
  Available at {\tt http://arxiv.org/abs/1104.2303v2}}.

\end{thebibliography}
\end{document}